\documentclass[12pt,onecolumn]{IEEEtran}
\usepackage{graphicx,amssymb,amsmath}
\usepackage[noadjust]{cite}
\usepackage{setspace}               
\usepackage{stfloats}

\usepackage{amsthm}
\usepackage{amsmath}
\usepackage{flushend}
\usepackage{cases,subeqnarray}
\usepackage{bm,multirow,bigstrut}
\usepackage{textcomp}
\usepackage{latexsym,bm}
\usepackage{booktabs,changebar}
\usepackage{xcolor}
\usepackage{mathtools}
\usepackage{dsfont}
\usepackage{extarrows}
\usepackage{mathrsfs}
\usepackage{cite}
\usepackage{bm}
\usepackage{cleveref}
\usepackage{multicol}       
\usepackage{multirow}       
\usepackage{array}          
\definecolor{crimson}{RGB}{192,0,0}         
\definecolor{navy}{RGB}{47,85,151}         

\makeatletter                               
\newif\if@restonecol
\makeatother

\makeatletter
\newif\if@restonecol
\makeatother

\usepackage[linesnumbered,ruled,vlined]{algorithm2e}
\usepackage{algpseudocode}
\usepackage{amsmath}
\usepackage{listings}                       
\usepackage[framed,numbered,autolinebreaks,useliterate]{mcode}
\theoremstyle{plain}
\newtheorem{thm}{Theorem}

\newtheorem{coro}{Corollary}

\theoremstyle{plain}
\newtheorem{rem}{Remark}

\IEEEoverridecommandlockouts

\begin{document}

\title{Cell-Free Massive MIMO with Multi-Antenna Users over Weichselberger Rician Channels
\thanks{X. Li, J. Zhang, and Z. Wang are with the School of Electronic and Information Engineering and the Frontiers Science Center for Smart High-speed Railway System, Beijing Jiaotong University, Beijing 100044, China. (e-mail: \{19211271, jiayizhang, zhewang\_77\}@bjtu.edu.cn).}
\thanks{B. Ai is with the State Key Laboratory of Rail Traffic Control and Safety, Beijing Jiaotong University, Beijing 100044, China, and also with the Frontiers Science Center for Smart High-speed Railway System, and also with Henan Joint International Research Laboratory of Intelligent Networking and Data Analysis, Zhengzhou University, Zhengzhou 450001, China, and also with Research Center of Networks and Communications, Peng Cheng Laboratory, Shenzhen, China (e-mail: boai@bjtu.edu.cn).}
\thanks{D. W. K. Ng is with School of Electrical Engineering and Telecommunications, University of New South Wales, NSW 2052, Australia (e-mail: w.k.ng@unsw.edu.au).}}
\author{Xin~Li, Jiayi~Zhang,~\IEEEmembership{Senior Member,~IEEE,} Zhe~Wang, \\ Bo~Ai,~\IEEEmembership{Fellow,~IEEE,} and Derrick Wing Kwan Ng,~\IEEEmembership{Fellow,~IEEE}}
\maketitle

\begin{abstract}
We consider a cell-free massive multiple-input multiple-output (MIMO) system with multi-antenna access points and user equipments (UEs) over Weichselberger Rician fading channels with random phase-shifts. More specifically, we investigate the uplink spectral efficiency (SE) for two pragmatic processing schemes: 1) the fully centralized processing scheme with global minimum mean square error (MMSE) or maximum ratio (MR) combining; 2) the large-scale fading decoding (LSFD) scheme with local MMSE or MR combining. To improve the system SE performance, we propose a practical uplink precoding scheme based on only the eigenbasis of the UE-side correlation matrices. Moreover, we derive novel closed-form SE expressions for characterizing the LSFD scheme with the MR combining. Numerical results validate the accuracy of our derived expressions and show that the proposed precoding scheme can significantly improve the SE performance compared with the scenario without any precoding scheme.
\end{abstract}


\IEEEpeerreviewmaketitle

\section{Introduction}

As one of the most promising technologies for enabling future communications, cell-free massive multiple-input multiple-output (CF mMIMO) has attracted extensive attention recently \cite{smallcell,9622183,Zhang2020JSAC,9650567}. In CF mMIMO systems, a large number of access points (APs), all connected to a central processing unit (CPU), are geographically distributed and jointly serve the user equipments (UEs) via the same time-frequency resources \cite{9622183}. In fact, since the CF mMIMO systems can overcome the inter-cell interference in cellular networks and obtain high macro-diversity gain, it can greatly improve the spectral efficiency, energy efficiency, and service quality of UEs \cite{8097026}.

The vast majority of existing papers revealed the performance of CF mMIMO with single-antenna UEs, e.g. \cite{Making,HIzheng,performance}. However, UEs in the fifth-generation (5G) communication networks have already been equipped with multiple antennas to exploit the spatial degrees of freedom for improving the system performance. The authors of \cite{multiUE3} analyzed the downlink SE performance of CF mMIMO systems with/without the downlink pilot transmission. The authors of \cite{8901451} investigated a CF mMIMO system with a user-centric approach for multi-antenna UEs. Besides, the authors of \cite{multiUE2} analyzed the impact of the number of antennas per UE on the performance of CF mMIMO. Also, the authors of \cite{8811486} investigated a CF mMIMO system with multi-antenna UEs and low-resolution ADCs. However, these works are based on the overly idealistic and simple assumption of uncorrelated Rayleigh fading channels with limited practical applications. As a remedy, a more realistic channel model for multi-antenna UEs was proposed in \cite{Weichselberger}, known as the Weichselberger model, which can capture the joint-correlation dynamics between the AP and UE-side. Inspired by \cite{Weichselberger}, the UL performance of CF mMIMO systems over the Weichselberger Rayleigh model was first investigated in \cite{Zhe_TWC}.

Unfortunately, many of the existing works, e.g. \cite{HIzheng,Making,Zhe_TWC}, ignored the existence of the line-of-sight (LoS) path, which is the dominant channel component of practical CF mMIMO systems. A more practical spatially correlated Rician fading channel was investigated in \cite{performance} with both single-antenna APs and UEs. This considered channel is composed of a semi-deterministic LoS path component with random phase-shifts and a stochastic non-line-of-sight (NLoS) path component. Moreover, based on \cite{performance}, the authors of \cite{uplink_performance} analyzed the UL performance of CF mMIMO with multi-antenna APs focusing on single-antenna UEs. Yet, the obtained results from \cite{uplink_performance} are not applicable to the case of multi-antenna UEs commonly adopted nowadays due to the setting of single-antenna UEs.

Motivated by the aforementioned observations, we investigate a CF mMIMO system with multi-antenna APs and UEs over Weichselberger Rician fading channels with random phase-shifts. The major contributions of this paper are listed as follows.

\begin{itemize}
\item For the first time, we investigate the CF mMIMO systems over Weichselberger Rician fading channels with random phase-shifts in the LoS component for two practical processing schemes: 1) the fully centralized processing scheme with global minimum mean square error (MMSE) or maximum ratio (MR) combining; 2) the large-scale fading decoding (LSFD) scheme with local MMSE or MR combining.
    Moreover, novel and exact closed-form SE expressions are computed for the LSFD scheme with the MR combining.
\item To improve the system SE performance, we propose a heuristic uplink precoding scheme based on only the eigenbasis of the UE-side correlation matrices. The proposed uplink precoding scheme is shown to improve the SE performance effectively.
\end{itemize}

\textbf{\emph{Notation}}:
$\mathbb {E} \{ {\cdot} \} $ represents the expectation and $\| \cdot \|_{\mathrm{F}}^2$ denotes the squared Frobenius norm. The Kronecker product and the Hadamard product are denoted by $ \otimes $ and $ \odot $. ${\mathrm{vec}} (\bf X )$ is the column vector obtained by stacking the columns of a matrix $\bf X$. ${\mathcal{N}}_{\mathbb C}({\bf {0},\bf{R}})$ denotes the circularly symmetric complex Gaussian distribution with zero mean and the correlation matrix ${\bf R}$.


\section{System Model}

As illustrated in Fig. \ref{System_Model}, we consider a CF mMIMO system, where $M$ APs and $K$ UEs are randomly located within a large area with $L$ and $N$ being the number of antennas per AP and UE, respectively. We investigate the standard block fading model, where the channel responses remain constant in each coherence time-frequency block \cite{Network}. Specifically, each block has ${\tau _c}$ samples, where ${\tau _p}$ and ${\tau _u = \tau_c - \tau_p}$ samples are reserved for UL training and data transmission, respectively. Let ${{\bf{H}}_{mk}}\in {\mathbb{C}^{L \times N}}$ denote the complex-valued channel between the $k$-th UE and the $m$-th AP. We assume that ${{\bf{H}}_{mk}}$ are independent for different AP-UE pairs and ${{\bf{H}}_{mk}}$ in different blocks are independent and identically distributed (i.i.d.).

\begin{figure}[t]
\centering
\includegraphics[scale=0.8]{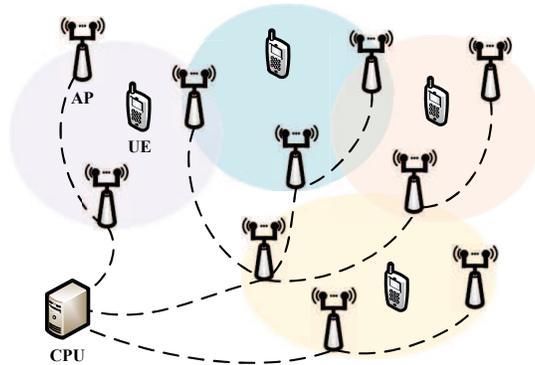}
\centering
\caption{A cell-free massive MIMO network.}\label{System_Model}
\vspace{-0.45cm}
\end{figure}

\subsection{Weichselberger Rician Fading Channel}

We consider the jointly-correlated Rician fading channel that is composed of a semi-deterministic LoS path component with a random phase-shift and a stochastic NLoS path component. Based on the Weichselberger model \cite{Weichselberger}, \cite{Zhe_TWC}, ${{\bf{H}}_{mk}}$ is modeled as
\begin{equation}
{{\bf{H}}_{mk}} = {\bar {\bf{H}} _{mk}}{e^{j{\varphi _{mk}}}} + \underset{\tilde{\bf{H}}_{mk}}{\underbrace{{{\bf{U}}_{mk,\text{r}}}\left( {{{\tilde {\bf{W}}}_{mk}} \odot {{\bf{H}}_{mk,\text{iid}}}} \right){\bf{U}}_{mk,\text{t}}^H}},
\end{equation}
where ${\bar {\bf{H}} _{mk}}\!=\! [{{\bar {\bf{h}} }_{mk1}}, \cdots ,{{\bar {\bf{h}} }_{mkN}}]$ is the deterministic LoS component with ${{\bar {\bf{h}} }_{mkn}}\in {\mathbb{C}^{L}}$ being the channel between the $n$-th antenna of UE $k$ and AP $m$. ${\varphi _{mk}}\! \sim \!\mathcal{U}[ { - \pi ,\pi } ]$ is the uniformly distributed phase-shift of the LoS component. Note that we assume that the phase-shifts from all the antennas of UE $k$ to AP $m$ are identical, and ${\varphi _{mk}}$ in different coherence blocks are i.i.d. \cite{performance,uplink_performance}. ${\tilde {\bf{H}}_{mk}}$ denotes the stochastic NLoS component and ${{\bf{H}}_{mk,\text{iid}}} \!\in\! {\mathbb{C}^{L \times N}}$ is composed of i.i.d. $\mathcal{N}_{\mathbb{C}}( {0,1} )$ entries. The unitary matrices ${{\bf{U}}_{mk,\text{r}}} \!\in\! {\mathbb{C}^{L \times L}}$ and ${{\bf{U}}_{mk,\text{t}}} \in {\mathbb{C}^{N \times N}}$ are the eigenbasis of the one-sided correlation matrices ${{\bf{R}}_{mk,\text{r}}} \buildrel \Delta \over = \!\mathbb{E}\{ {{{\tilde {\bf{H}}}_{mk}}\tilde {\bf{H}}_{mk}^H} \}$ and ${{\bf{R}}_{mk,\text{t}}} \buildrel \Delta \over=\! \mathbb{E}\{ {\tilde {\bf{H}}_{mk}^T{{\tilde {\bf{H}}}_{mk}}^*} \}$, respectively \cite{Weichselberger}. Moreover, the eigenmode coupling matrix is defined as ${{\bf{W}}_{mk}} \!\buildrel \Delta \over =\! {\tilde {\bf{W}}_{mk}} \odot {\tilde {\bf{W}}_{mk}}$, capturing the spatial arrangement of scattering objects, whose the $(i,j)$-th element $[{\bf{W}}_{mk}]_{ij}$ specifies the average amount of power coupling from the $i$-th column of ${\bf U}_{mk,\text{r}}$ to the $j$-th column of ${\bf U}_{mk,\text{t}}$.

Let ${{\bf{h}}_{mk}}\! =\! {\mathrm{vec}} ( {{{\bf{H}}_{mk}}})\!\in\! {\mathbb C}^{LN }$ and the full correlation matrix
${{\bf{R}}_{mk}}\! \buildrel \Delta \over=\! \mathbb{E}\{ {{\mathrm{vec}} ( {{{\tilde {\bf{H}}}_{mk}}} ){\mathrm{vec}} {{( {{{\tilde {\bf{H}}}_{mk}}} )}^H}}\}\!\in\! {\mathbb C}^{LN \times LN}$ is given by ${{\bf{R}}_{mk}} \!=\! ( {{{\bf{U}}_{mk,\text{t}}^*}\! \otimes\! {{\bf{U}}_{mk,\text{r}}}} ){\mathrm{diag}} ( {{\mathrm{vec}} ( {{{\bf{W}}_{mk}}} )} ){( {{{\bf{U}}_{mk,\text{t}}^*}\! \otimes\! {{\bf{U}}_{mk,\text{r}}}})^H}$. Moreover, we can rewrite ${{\bf{R}}_{mk}}$ into the concatenation of ${N^2}$ blocks with the $(i,j)$-th submatrix being ${\bf{R}}_{mk}^{ij} = \mathbb{E}\{ {\bf{h}}_{mki}{\bf{h}}_{mkj}^H\}-{{\bar {\bf{h}} }_{mki}}{{\bar {\bf{h}} }_{mkj}^H}$ with ${\bf{h}}_{mki}$ and ${\bf{h}}_{mkj}$ being the $i$-th column and $j$-th column of ${{\bf{H}}_{mk}}$, respectively.

\begin{rem}
The motivations for adopting the Weichselberger Rician channel model in this paper are: 1) The Weichselberger Rician fading channel model with random phase-shifts in the LoS component is more practical than the commonly investigated Rayleigh fading channel model; 2) The Weichselberger model is realistic and suitable for practical scenarios with multi-antenna UEs. In particular, the Weichselberger model can capture the joint-correlation dynamics between the AP and UE-side \cite{Zhe_TWC}. More importantly, the eigenmode coupling matrix in the model can characterize the spatial arrangement of the scattering objects between each AP-UE pair in practical environment. In other words, the Weichselberger model embraces most channels of great interest. Moreover, the measurement campaigns in \cite{Weichselberger} validated that the Weichselberger model shows significantly less modeling error than other popular stochastic channel models.
\end{rem}

\subsection{Uplink Transmission}
\subsubsection{Channel Estimation}

We adopt mutually orthogonal pilot matrices for channel estimation and each pilot matrix is composed of $N$ mutually orthogonal pilot sequences \cite{multiUE3}. Let ${{\bf{P}}_k} \in {\mathbb{C}^{N\times {\tau _p}}}$ denote the pilot matrix of UE $k$ such that $\mathbf{P}_k\mathbf{P}_{k^{\prime}}^{H}=\tau _p\mathbf{I}_N$, if $\ k^{\prime}=k$ and ${\bf{0}}$ otherwise. We investigate a practical case where more than one UE exploits the same pilot matrix due to the limited system resources \cite{Making}. Let ${{\mathcal{P}}_k}$ denote the index subset of UEs that adopt the same pilot matrix as UE $k$ including itself. All UEs send their pilot signals and the received signal ${{\bf{Y}}_m^p}\in {\mathbb{C}^{L \times {\tau _p}}}$ at AP $m$ is
\begin{equation}
{{\bf{Y}}_{m}^{p}} = \sum\nolimits_{k = 1}^K {{{\bf{H}}_{mk}}} {{\bf{F}}_k}{{\bf{P}}_k} + {{\bf{N}}_m}, 
\end{equation}
where ${{\bf{F}}_k}\in {\mathbb{C}^{N \times N}}$ is the uplink precoding matrix of UE $k$ that satisfies $\mathbb{E}\{ {\| {{{\bf{F}}_k}} \|_{\mathrm{F}}^2} \} \leq {p_k}$ with ${p_k}$ being the maximum transmit power of UE $k$ and ${{\bf{N}}_m} \in {\mathbb{C}^{L \times {\tau _p}}}$ is the noise matrix with $\mathcal{N}_{\mathbb{C}}( {0,{\sigma ^2}} )$ elements.

To obtain the estimate of ${{\bf{h}}_{mk}}$, AP $m$ correlates ${\bf{Y}}_{m}^p$ with ${\bf{P}}_k^H$ as
\begin{equation}
{{\bf{Y}}}_{mk}^{p} = {{\bf{Y}}_m^p}{{\bf{P}}_k^H}= \sum\nolimits_{l \in {{\mathcal{P}}_k}} {{\tau _p}{{\bf{H}}_{ml}}{{\bf{F}}_l}}  + {{\bf{N}}_m}{\bf{P}}_k^H. 
\end{equation}
We assume that ${\varphi _{mk}}$ is available at each AP and derive the phase-aware MMSE estimate of ${{\bf{h}}_{mk}}$ as
\begin{equation}
{\hat {\bf{h}}}_{mk}\!=\!{\mathrm{vec}}( {{\hat {\bf{H}}}_{mk}} )\!=\!{{\bar{\bf{h}}}_{mk}} {e^{j{\varphi _{mk}}}} \!+\! {{\bf{R}}_{mk}}{\tilde {\bf{F}}_k}{\bf{\Psi}} _{mk}^{ - 1}{\mathrm{vec}}( {{\bf{Y}}_{mk}^p} ), 
\end{equation}
where ${{\bar{\bf{h}}}_{mk}}={\mathrm{vec}}( {{{\bar {\bf{H}} }_{mk}}} )$, ${\tilde {\bf{F}}_k} = {\bf{F}}_k^T \otimes {{\bf{I}}_L}$, and
${{\bf{\Psi}} _{mk}} = \sum\nolimits_{l \in {{\mathcal{P}}_k}} {{\tau _p}{{\tilde{\bf{F}}}_l}{{\bf{R}}_{ml}}\tilde {\bf{F}}_l^H} \! + {\sigma ^2}{{\bf{I}}_{LN}}$, respectively. The estimate ${\hat {\bf{h}}_{mk}}$ and estimation error ${{\bf{e}}_{mk}}\!=\!{{\bf{h}}_{mk}} - {\hat {\bf{h}}_{mk}}$ are independent random vectors with
\begin{equation}\notag
\begin{split}
{\mathbb{E}}\{ { {\hat {\bf{h}}_{mk}} \left| {{e^{j{\varphi _{mk}}}}} \right. } \} = {\bar {\bf{h}}_{mk}}{e^{j{\varphi _{mk}}}},&{\mathrm{Cov}}\{ {{\hat {\bf{h}}_{mk}} \left| {{e^{j{\varphi _{mk}}}}} \right.} \} ={{\bf{{\bf{\Phi}} }} _{mk}},\\
{\mathbb{E}}\{ { \bf{e}}_{mk}  \} = {\bf 0},\quad \quad \ \ \ \ \ \ \ \ \ \ \ \ \ \ &{\mathrm{Cov}}\{ {{ {\bf{e}}_{mk}}} \} ={{\bf C} _{mk}},
\end{split}
\end{equation}
where ${{\bf{{\bf{\Phi}} }} _{mk}} = {\tau _p}{{\bf{R}}_{mk}}{\tilde {\bf{F}}_k}{\bf{\Psi }} _{mk}^{ - 1}\tilde {\bf{F}}_k^H{{\bf{R}}_{mk}}$ and ${{\bf{C }} _{mk}} = {{\bf{R}}_{mk}}-{{\bf{{\bf{\Phi}} }} _{mk}}$.

\subsubsection{Data Transmission}

During the phase of data transmission, all antennas of all the UEs transmit their data symbols to all APs. The received signal of AP $m$ is
\begin{equation}
{{\bf{y}}_m} = \sum\nolimits_{k = 1}^K {{{\bf{H}}_{mk}}} {{\bf{s}}_k} + {{\bf{n}}_m},
\end{equation}
where ${{\bf{s}}_k} = {{\bf{F}}_k}{{\bf{x}}_k}\in {\mathbb{C}^{N}}$ is the transmitted signal of UE $k$ with ${{\bf{x}}_k} \sim \mathcal{N}_{\mathbb{C}}( {{\bf 0},{{\bf{I}}_N}} )$ being the data symbol vector of UE $k$ and ${{\bf{n}}_m} \sim \mathcal{N}_{\mathbb{C}}( {{\bf 0},{\sigma ^2}{{\bf{I}}_L}} )$ is the additive noise vector. We assume that ${{\bf{F}}_k}$ is designed based on only the channel statistics so is available at all APs and the CPU. Later, we will propose an effective precoding method to improve the system performance.

\section{Spectral Efficiency Analysis}

Depending on the required performance and complexity tradeoff in principle, four signal processing schemes can be implemented in CF mMIMO \cite{Making}. In particular, the ``fully centralized processing'' and the ``LSFD processing'' are regarded as competitive schemes for utilizing the potential of CF mMIMO. In this section, we investigate these two processing schemes and analyze their corresponding SE performance.

\subsection{Fully Centralized Processing}

Under the setting of the ``fully centralized processing'', APs forward all the received pilot signals and data signals to the CPU. Indeed, all processing is implemented at the CPU. The collective channel of UE $k$ is ${{\bf{h}}_k}\! =\! {[ { {{{\bf{h}}_{1k}^T}}, \cdots ,{{{\bf{h}}_{Mk}^T}}}  ]^T}\in {{\mathbb{C}}^{MLN}}$ with the mean ${\bar {\bf{h}}_k} \!=\! [ {{\bar {\bf{h}}_{1k}^T} {e^{j{\varphi _{1k}}}}, \cdots ,{\bar {\bf{h}}_{Mk}^T}{e^{j{\varphi _{Mk}}}}} ]^T$ and the full convariance matrix ${\bf{R }}_k \!=\! {\mathrm{diag}} ( {{\bf{R }}_{1k}, \cdots ,{\bf{R}}_{Mk}} )$. For the fully centralized processing, we assume that ${\varphi _{mk}}$ is available at the CPU so channel estimates can be derived by ``phase-aware'' MMSE estimators at the CPU. For UE $k$, the collective channel estimate can be constructed as ${\hat {\bf{h}}_k} \!=\!{[ { {{{\hat {\bf{h}}}_{1k}^T}}, \cdots ,{{{\hat{\bf{h}}}_{Mk}^T}}}  ]^T}$ with ${\mathbb{E}}\{ { {\hat {\bf{h}}_{k}} | {{e^{j{\varphi _{k}}}}}  } \}\! =\! {\bar {\bf{h}}_{k}}$, ${\mathrm{Cov}}\{ {{\hat {\bf{h}}_{k}} | {{e^{j{\varphi _{k}}}}} } \} \!=\!{\tau _p}{{\bf{R}}_{k}}{\bar {\bf{F}}_k}{\bf{\Psi}} _{k}^{ - 1}\bar {\bf{F}}_k^H{{\bf{R}}_{k}}$, where ${\bf{\Psi }}_k^{ - 1} \!=\! {\mathrm{diag}} ( {{\bf{\Psi }}_{1k}^{ - 1}, \cdots ,{\bf{\Psi }}_{Mk}^{ - 1}})$ and ${\bar {\bf{F}} _k} \!=\! {\mathrm{diag}} (\, {\underset{M}{\underbrace{{{{\tilde {\bf{F}}}_k}, \cdots ,{{\tilde {\bf{F}}}_k}}}}} \,)$.

In addition, the received signal available at the CPU is
\begin{equation}
\underset{=\mathbf{y}}{\underbrace{\left[ {\begin{array}{*{20}{c}}
{{{\bf{y}}_1}}\\
 \vdots \\
{{{\bf{y}}_M}}
\end{array}} \right]}} = \sum\limits_{k = 1}^K {\underset{={\mathbf{H}}_k}{\underbrace{\left[ {\begin{array}{*{20}{c}}
{{{\bf{H}}_{1k}}}\\
 \vdots \\
{{{\bf{H}}_{Mk}}}
\end{array}} \right]}}{{\bf{F}}_k}{{\bf{x}}_k}}  + \underset{=\mathbf{n}}{\underbrace{\left[ {\begin{array}{*{20}{c}}
{{{\bf{n}}_1}}\\
 \vdots \\
{{{\bf{n}}_M}}
\end{array}} \right]}}.
\end{equation}
Besides, the received signal can be expressed as ${\bf{y}} = \sum\nolimits_{k = 1}^K {{{\bf{H}}_k}{{\bf{F}}_k}{{\bf{x}}_k}}  + {\bf{n}}$. Furthermore, an arbitrary receive combining matrix ${\bf{V}}_k\in \mathbb{C}^{LM \times N}$ can be designed by the CPU based on the collective channel estimates to decode ${\bf x}_k$ as
\begin{equation}
{\hat {\bf{x}}_k} = {\bf{V}}_k^H{ {\bf{H}}_k}{{\bf{F}}_k}{{\bf{x}}_k} + \sum\limits_{l \ne k}^K {{\bf{V}}_k^H{{ {\bf{H}}}_l}{{\bf{F}}_l}{{\bf{x}}_l}}  + {\bf{V}}_k^H{\bf{n}}.
\end{equation}
Based on ${\hat {\bf{x}}_k}$, we can derive the achievable SE for UE $k$ by using standard capacity lower bounds \cite{Making} as summarized in the following corollary.

\begin{coro}\label{coro 1}
For the fully centralized processing with a given ${\bf{V}}_k$, the achievable SE for UE $k$ with the phase-aware MMSE estimator is given by
\begin{equation}\label{L4 SE 1}
{{\mathrm {SE}}_k^{(1)}} = \left( {1 - \frac{\tau _p}{\tau _c}} \right){\mathbb{E}} \left\{ {\log _2}\left| {{{\bf{I}}_N} +
{\bf{D}}_{k,(1)}^H{\bf{\Sigma }}_{k,(1)}^{ - 1}{{\bf{D}}_{k,(1)}}
} \right| \right\},
\end{equation}
where ${{\bf{D}}_{k,(1)}}\!\! \buildrel \Delta \over =\! {\bf{V}}_k^H{\hat {\bf{H}}_k}{{\bf{F}}_k}$ and ${{\bf{\Sigma }}_{k,(1)}} \!\buildrel \Delta \over =\! {\bf{V}}_k^H( {\sum\nolimits_{l = 1}^K {{{\hat {\bf{H}}}_l}{{\hat{\bf{F}}}_l}{\hat {\bf{H}}}_l^H}\!+ \!\sum\nolimits_{l=1}^K{{\bf{C}}_l^{\prime}}\!+\!{\sigma^2}{{\bf{I}}_{ML}}}){{\bf{V}}_k}\!-\!{{\bf{D}}_{k,{(1)}}}{\bf{D}}_{k,{(1)}}^H$ with ${{\hat{\bf F}}_l}={{\bf{F}}_l}{\bf{F}}_l^H$ and ${{\bf{C}}_l^{\prime}} \!=\! {\mathrm{diag}}( {{\bf{C}}_{1l}^\prime , \cdots ,{\bf{C}}_{Ml}^\prime })$ with the $(i,j)\text{-th}$ element of ${\bf{C}}_{ml}^\prime$ being ${[ {{\bf{C}}_{ml}^\prime } ]_{ij}} = \sum\nolimits_a^N {\sum\nolimits_b^N {{{[ {\hat{\bf F}}_l ]}_{ba}}} } {[ {{\bf{C}}_{ml}^{ba}} ]_{ij}}$.
\end{coro}
\begin{proof}
The proof follows from \cite[Corollary 1]{Zhe_TWC}, thus is omitted for brevity.
\end{proof}

In this paper, we consider two promising combining schemes: the MR combining ${{\bf{V}}_{k}} = {\hat {\bf{H}}_{k}}$ and the MMSE combining given by
\begin{equation}\label{L4 MMSE combing}
{{\bf{V}}_{k}} \!=\! {\left( {\sum\limits_{l = 1}^K \left({{\hat{\bf{H}}}_l {{\hat{\bf F}}_l}{\hat{\bf{H}}}_l^H \!+\! {\bf{C}}_l^{\prime}} \right) \!+\! {\sigma ^2}{{\bf{I}}_{ML}}} \right)^{ - 1}}{{\hat{\bf{H}}}_{k}}{{\bf{F}}_k},
\end{equation}
which can minimize ${{\mathrm {MSE}}_{k}} = {\mathbb E}\{ {{{\| {{{\bf{x}}_k} - {\bf{V}}_{k}^H{\bf{y}}} \|}^2}| {{{{\hat {\bf{H}}}_{k}}} } } \}$ \cite{Zhe_TWC}. Note that the MMSE combining can also maximize \eqref{L4 SE 1} as the following corollary.

\begin{coro}\label{coro 2}
The MMSE combining matrix in \eqref{L4 MMSE combing} can maximize \eqref{L4 SE 1} with the maximum value as ${{\mathrm {SE}}_k^{(1)}} = ( {1 - \frac{{{\tau _p}}}{{{\tau _c}}}})\mathbb E \{ {\log _2}| {{{\bf{I}}_N} +(\mathbf{D}_{k,( 1 )}^{\prime})^H(\mathbf{\Sigma}_{k,(1)}^{\prime})^{-1} \mathbf{D}_{k,(1)}^{\prime}} |\}$ with $\mathbf{D}_{k,( 1)}^{\prime}\buildrel \Delta \over =\mathbf{\hat{H}}_k\mathbf{F}_k$ and $\mathbf{\Sigma }_{k,(1)}^{\prime}\buildrel \Delta \over =\sum_{l\ne k}^K{\mathbf{\hat{H}}_l\mathbf{\hat{F}}_l\mathbf{\hat{H}}_{l}^{H}} +\sum_{l=1}^K {\!\mathbf{C}_{l}^{\prime}}+\sigma ^2\mathbf{I}_{ML}$.
\end{coro}
\begin{proof}
Please refer to Appendix \ref{appendix B}.
\end{proof}

\subsection{Large-Scale Fading Decoding}

In this subsection, we investigate a two-layer processing scheme with the local processing at each AP as the first layer and the LSFD method at the CPU as the second layer \cite{Making,performance}. Let ${\bf{V}}_{mk} \in {\mathbb{C}}^{L\times N}$ denote the local combining matrix that AP $m$ selects for UE $k$. At AP $m$, the local estimate of ${{\bf{x}}_{k}}$ is
\begin{equation}
{\tilde {\bf{x}}_{mk}} \!=\! {\bf{V}}_{mk}^H{{\bf{H}}_{mk}}{{\bf{F}}_k}{{\bf{x}}_k} \!+\! \sum\limits_{l \ne k}^K {{\bf{V}}_{mk}^H{{\bf{H}}_{ml}}{{\bf{F}}_l}{{\bf{x}}_l}}  \!+\! {\bf{V}}_{mk}^H{{\bf{n}}_m}.
\end{equation}

Note that ${{\bf{V}}_{mk}}$ is designed at each AP by exploiting the local CSI. One possible scheme is the MR combining ${{\bf{V}}_{mk}} = {\hat {\bf{H}}_{mk}}$ and another handy choice is local MMSE (L-MMSE) combining ${{\bf{V}}_{mk}} = {( {\sum\nolimits_{l = 1}^K ({{\hat{\bf{H}}}_{ml} {{\hat{\bf F}}_l} {\hat{\bf{H}}}_{ml}^H + {\bf{C}}_{ml}^{\prime}} ) + {\sigma ^2}{{\bf{I}}_{L}}} )^{ - 1}}{{\hat{\bf{H}}}_{mk}}{{\bf{F}}_k}$.

The local estimates are then converged to the CPU, where they are linearly combined by the LSFD coefficient matrix ${{\bf{A}}_{mk}}\!\in\! {\mathbb C}^{N\times N}$ to obtain the final decode ${\mathbf{\bar{x}}_k}$ as ${\mathbf{\bar{x}}_k}  \!=\! \sum\nolimits_{m = 1}^M {{\bf{A}}_{mk}^H{{\tilde {\bf{x}}}_{mk}}}$. Let ${{\bf{A}}_k}\! =\!\! [{\bf{A}}_{1k}^H, \cdots, {\bf{A}}_{Mk}^H ]$, ${{\bf{G}}_{kl}}\!  =\!\! [ ({\bf{V}}_{1k}^H{{\bf{H}}_{1l}})^T, \cdots ,({\bf{V}}_{Mk}^H{{\bf{H}}_{Ml}})^T ]^T$ and ${{\bf{n}}_k^{\prime}} \! =\! [({\bf{V}}_{1k}^H{{\bf{n}}_1})^T, \cdots ,({\bf{V}}_{Mk}^H{{\bf{n}}_M})^T]^T$. The decoded signal can be written as
\begin{equation}
{\mathbf{\bar{x}}_k} = {{\bf{A}}_k}{{\bf{G}}_{kk}}{{\bf{F}}_k}{{\bf{x}}_k} + \sum\nolimits_{l \ne k}^K {{{\bf{A}}_k}{{\bf{G}}_{kl}}{{\bf{F}}_l}{{\bf{x}}_l}}  + {{\bf{A}}_k}{{\bf{n}}_k^{\prime}}.
\end{equation}
Note that since only the channel statistics are available at the CPU, the classical use-and-then-forget (UatF) bound \cite{Network} is applied to derive the following ergodic achievable SE.

\begin{coro}\label{coro 3}
The achievable SE for UE $k$ is given by
\begin{equation}\label{L3 SE 1}
{{\mathrm {SE}}_k^{\left(2 \right)}} = \left( {1 - \frac{{{\tau _p}}}{{{\tau _c}}}} \right){\log _2}\left| {{{\bf{I}}_N} + {\bf{D}}_{k,{\left(2 \right)}}^H{\bf{\Sigma }}_{k,{\left(2 \right)}}^{ - 1}{{\bf{D}}_{k,{\left(2 \right)}}}} \right|,
\end{equation}
where ${{\bf{D}}_{k,{(2 )}}} \buildrel \Delta \over = {{\bf{A}}_k}{\mathbb E}\{ {{{\bf{G}}_{kk}}} \}{{\bf{F}}_k}$ and $\!{{\bf{\Sigma }}_{k,{(2)}}} \!\!\buildrel \Delta \over  = \!\!\sum\nolimits_{l = 1}^K {{{\bf{A}}_k} {\mathbb E}\{ {{{\bf{G}}_{kl}}{{\hat{\bf F}}_l}{{\bf{G}}_{kl}^H}} \}{\bf{A}}_k^H} \! +\! {\sigma ^2}{{\bf{A}}_k}  {{{\bf{S}}_k}} {\bf{A}}_k^H\! -\! {{\bf{D}}_{k,{(2 )}}}{\bf{D}}_{k,{(2)}}^H$ with ${{{\bf{S}}_k}}\!\!=\! {\mathrm{diag}}( {{\mathbb E}\{ {{\bf{V}}_{1k}^H{{\bf{V}}_{1k}}} \}, \!\cdots\! ,{\mathbb E}\{ {{\bf{V}}_{Mk}^H{{\bf{V}}_{Mk}}} \}} ) $.
\end{coro}
\begin{proof}
The proof is similar to Corollary \ref{coro 1}, thus is omitted.
\end{proof}

Note that ${\bf{A}}_k$ can be optimized by the CPU for maximizing the achievable SE in \eqref{L3 SE 1} as the following corollary.

\begin{coro}
The achievable SE in \eqref{L3 SE 1} can be maximized by
\begin{equation}\label{L3 Ak}
{{\bf{A}}_k} = {\left( {\sum\limits_{l = 1}^K {{\mathbb E}\left\{ {{{\bf{G}}_{kl}}{{\hat{\bf F}}_l}{{\bf{G}}_{kl}^H}} \right\}}  +  {\sigma ^2}{{{\bf{S}}_k}}} \right)^{ - 1}}{\mathbb E} \left\{ {{{\bf{G}}_{kk}}} \right\}{{\bf{F}}_k},
\end{equation}
which leads to the maximum value ${{\mathrm {SE}}_k^{(2 )}} \!=\! ( {1 - \frac{{{\tau _p}}}{{{\tau _c}}}} ){\log _2}| {
{{\bf{I}}_N} \!+\!
( \mathbf{D}_{k,( 2 )}^{\prime} ) ^H( \mathbf{\Sigma }_{k,( 2 )}^{\prime} ) ^{-1}\mathbf{D}_{k,( 2 )}^{\prime}} |$ with ${{\bf{D}}_{k,{(2 )}}^{\prime}} \buildrel \Delta \over = {\mathbb E}\{ {{{\bf{G}}_{kk}}} \}{{\bf{F}}_k}$ and ${{\bf{\Sigma }}_{k,{(2 )}}^{\prime}}\buildrel \Delta \over  = \sum\nolimits_{l = 1}^K{\mathbb E}\{ {{{\bf{G}}_{kl}}{{\hat{\bf F}}_l}{{\bf{G}}_{kl}^H}} \} + {\sigma ^2}  {{{\bf{S}}_k}}  - {{\bf{D}}_{k,{(2)}}^{\prime}}( \mathbf{D}_{k,( 2 )}^{\prime} ) ^H$.
\end{coro}
\begin{proof}
The proof is similar to Corollary \ref{coro 2}, thus is omitted.
\end{proof}
Furthermore, when the MR combining is applied at each AP, we can compute the expectations in \eqref{L3 SE 1} as closed-form and derive the closed-form SE expressions as the following theorem.

\begin{thm}\label{thm_1}
If the MR combining with ${{\bf{V}}_{mk}} = {\hat {\bf{H}}_{mk}}$ is adopted, we can obtain the closed-form SE expressions for UE $k$ as \footnote{Note that the phase-shifts in the LoS component are unaware at the CPU for the LSFD scheme so would have an impact on the derivation of SE. An illustration is the case with $l \in {{\mathcal P}_k}\backslash\{k\}$ as shown in Appendix~\ref{appendix C}.}
\begin{equation}\label{L3 SE 2}
{{\mathrm {SE}}_k^{\left(2 \right)}} = \left( {1 - \frac{{{\tau _p}}}{{{\tau _c}}}} \right){\log _2}\left| {
{{\bf{I}}_N} + {\bf{D}}_{k,{\left(2 \right)}}^H{\bf{\Sigma }}_{k,{\left(2 \right)}}^{ - 1}{{\bf{D}}_{k,{\left(2 \right)}}}  } \right|,
\end{equation}
\sloppy
where ${{\bf{D}}_{k,( 2 )}} \!\buildrel \Delta \over =\! {{\bf{A}}_k}{{\bf Z}_k}{{\bf{F}}_k}$ and ${{\bf{\Sigma }}_{k,( 2 )}} \!\buildrel \Delta \over =\! \sum\nolimits_{l = 1}^K {{{\bf{A}}_k}{{\bf{\Theta }}_{kl}}{\bf{A}}_k^H}  + {\sigma ^2}{{\bf{A}}_k}{{\bf{S}}_k}{\bf{A}}_k^H -{{\bf{D}}_{k,(2)}}{\bf{D}}_{k,(2)}^H$ with ${{\bf Z}_k}={[ {
{{\bf{Z}}_{1k}^T}, \cdots ,{{\bf{Z}}_{Mk}^T}} ]^T}$ and ${{\bf{S}}_k} \!=\! {\mathrm{diag}} ({{\bf{Z}}_{1k}}, \cdots ,{{\bf{Z}}_{Mk}})$ with $(i,j)$-th element of ${{{\bf{Z}}_{mk}}} \in {\mathbb{C}}^{N \times N}$ being ${[ {{{\bf{Z}}_{mk}}} ]_{ij}} = {\mathrm {tr}} ( {{\bf{{\bf{\Phi}} }} _{mk}^{ji} + {{\bar {\bf{h}} }_{mkj}}{\bar {\bf{h}}} _{mki}^H} )$. ${{{\bf{\Theta}} _{kl}}}\in \mathbb{C} ^{MN\times MN}$ can be structured into $M^2$ blocks with the $(m,n)$-submatrix ${\bf{\Theta }}_{kl}^{mn}\in \mathbb{C} ^{N\times N}$ being
\begin{equation}
{\bf{\Theta }}_{kl}^{mn} = \left\{ {\begin{array}{*{20}{c}}
{{\bf{\Gamma }}_{kl}^{mm} + {\bf{T}}_{kl}^m,m = n}\\
{{\bf{\Gamma }}_{kl}^{mn},\qquad\quad m \ne n}
\end{array}} \right.
\end{equation}
with the $(i,j)$-th element of ${\bf{\Gamma}}_{kl}^{mn}$ and ${{\bf{T}}_{kl}^{m}}$ being ${[ {{\bf{\Gamma}}_{kl}^{mn}} ]_{ij}} \!=\! \sum\nolimits_a^N {\sum\nolimits_b^N {{{[ {{\hat{\bf F}}_l} ]}_{ba}}} } \, \gamma $ and ${[ {{\bf{T}}_{kl}^{m}} ]_{ij}} \!=\! \sum\nolimits_a^N {\sum\nolimits_b^N {{{[ {{\hat{\bf F}}_l} ]}_{ba}}} } \, t $, respectively, where $\gamma$ and $t$ (we omit the subscript ``${kl,mn,ij,ab}$'' for brevity) are given in \eqref{L3 fig} at the top of the next page.

\sloppy
Besides, the LSFD coefficient matrix in \eqref{L3 Ak} can be written in closed-form as 
${{\bf{A}}_k} = {( {\sum\nolimits_{l = 1}^K {\bf{\Theta }}_{kl}   +  {\sigma ^2}{{{\bf{S}}_k}}} )^{ - 1}} {{\bf{Z}}_{mk}} {{\bf{F}}_k}$.
\end{thm}
\begin{proof}
Please refer to Appendix \ref{appendix C}.
\end{proof}

\begin{rem}
Note that all achievable SE expressions hold for any ${{\bf{R}}_{mk}}$. Besides, only SE expressions for the LSFD scheme with the MR combining can be computed in the closed-form.
\end{rem}

\newcounter{m2}
\begin{figure*}[t]
\normalsize
\setcounter{m2}{\value{equation}}
\vspace*{4pt}
\setcounter{equation}{15}
\begin{equation}\label{L3 fig}
\begin{split}
{\gamma} &=
\left\{ \begin{array}{l}
\sum\nolimits_{u_1}^N {\sum\nolimits_{u_2}^N { {{\mathrm {tr}} ( {\tilde {\bf{\Phi}} {{_{mk}^{u_1i}}}\tilde {\bf{\Phi}} _{ml}^{bu_1}} )  {\mathrm {tr}} ( {\tilde {\bf{\Phi}} {{_{nl}^{u_2a}}}\tilde {\bf{\Phi}} _{nk}^{ju_2}} )} } }, \qquad\qquad\qquad\qquad\qquad\qquad\qquad\qquad l \in {{\mathcal P}_k}\backslash \{ k\}, \\
\sum\nolimits_{u_1}^N {\sum\nolimits_{u_2}^N { {{\mathrm {tr}} ( {\tilde {\bf{\Phi}} {{_{mk}^{u_1i}}}\tilde {\bf{\Phi}} _{ml}^{bu_1}} )  {\mathrm {tr}} ( {\tilde {\bf{\Phi}} {{_{nl}^{u_2a}}}\tilde {\bf{\Phi}} _{nk}^{ju_2}} )} } }
\! +\! \sum\nolimits_{u_1}^N {{\mathrm {tr}} ( {\bar {\bf{h}} _{mki}^H{{\bar {\bf{h}} }_{mlb}}\tilde {\bf{\Phi}} {{_{nl}^{u_1a}}}\tilde {\bf{\Phi}} _{nk}^{ju_1}} ) }\\
\! +\! \sum\nolimits_{u_1}^N {{\mathrm {tr}} ( {\tilde {\bf{\Phi}} {{_{mk}^{u_1i}}}\tilde {\bf{\Phi}} _{ml}^{bu_1}\bar {\bf{h}} _{nla}^H{{\bar {\bf{h}} }_{nkj}}} )}
 \!+\! \bar {\bf{h}} _{mki}^H{\bar {\bf{h}} _{mlb}}\bar {\bf{h}} _{nla}^H{\bar {\bf{h}} _{nkj}},\qquad\qquad\qquad\qquad\quad\quad\quad l = k,\\
\end{array} \right.\\
 {t} &=\!
\begin{array}{l}
\!{\mathrm {tr}} ( {( {{{\bar {\bf{h}} }_{mkj}}\bar {\bf{h}} _{mki}^H + {\bf{\Phi}} _{mk}^{ji}} )\cdot {{\bf{C}}_{ml}^{ba} } })
\!+\! \sum\nolimits_{u_1}^N { {\bar {\bf{h}} _{mki}^H\tilde {\bf{\Phi}} _{ml}^{bu_1}\tilde {\bf{\Phi}} {{_{ml}^{u_1a}}}{{\bar {\bf{h}} }_{mkj}}} }
\!+\! \sum\nolimits_{u_1}^N {{\mathrm {tr}}( {\tilde {\bf{\Phi}} {{_{mk}^{u_1i}}}{{\bar {\bf{h}} }_{mlb}}\bar {\bf{h}} _{mla}^H\tilde {\bf{\Phi}} _{mk}^{ju_1}})}\!\!\\
\end{array}\\
&+ \left\{ \begin{array}{l}
\!\!{\mathrm {tr}}({\bf{\Phi}} _{mk}^{ji}{\bf{\Phi}} _{ml}^{ba})+\! \bar {\bf{h}} _{mki}^H{\bar {\bf{h}} _{mlb}}\bar {\bf{h}} _{mla}^H{\bar {\bf{h}} _{mkj}},\qquad\qquad\qquad\qquad\qquad\quad  l \notin {{\mathcal P}_k},\\
\!\!\sum\nolimits_{u_1}^N {\sum\nolimits_{u_2}^N { {{\mathrm {tr}} ( {\tilde {\bf{\Phi}} {{_{mk}^{u_1i}}}\tilde {\bf{\Phi}} _{ml}^{bu_2}\tilde {\bf{\Phi}} {{_{ml}^{u_2a}}}\tilde {\bf{\Phi}} _{mk}^{ju_1}} )} } }+\! \bar {\bf{h}} _{mki}^H{\bar {\bf{h}} _{mlb}}\bar {\bf{h}} _{mla}^H{\bar {\bf{h}} _{mkj}} ,\qquad  l \in {{\mathcal P}_k}\backslash \{ k\} ,\\
\!\!\sum\nolimits_{u_1}^N {\sum\nolimits_{u_2}^N { {{\mathrm {tr}} ( {\tilde {\bf{\Phi}} {{_{mk}^{u_1i}}}\tilde {\bf{\Phi}} _{ml}^{bu_2}\tilde {\bf{\Phi}} {{_{ml}^{u_2a}}}\tilde {\bf{\Phi}} _{mk}^{ju_1}} )}  } },\qquad\qquad\qquad\qquad\qquad\quad\,\,\, l = k.\\
\end{array} \right.
\end{split}
\end{equation}
\setcounter{equation}{\value{m2}}
\hrulefill
\vspace{-0.5cm}
\end{figure*}

\subsection{Precoding Matrix Design}

In order to improve the SE performance, the UL precoding can be implemented at the UE side in the scenario with multi-antenna UEs. Inspired by \cite{precoding_optimal}, for CF mMIMO, we propose a heuristic uplink precoding scheme based only on the correlation feature ${{\bf{U}}_{mk,\text{t}}}$ at the UE side as
\addtocounter{equation}{1}
\begin{equation}\label{F_precoding}
{{\bf{F}}_k} = \sqrt {{p_k}} \left\| {\sum\limits_{m = 1}^M {{{\bf{U}}_{mk,\text{t}}}} } \right\|_{\mathrm{F}}^{ - 1}\left( {\sum\limits_{m = 1}^M {{{\bf{U}}_{mk,\text{t}}}} } \right).
\end{equation}

\section{Numerical Results}

The $M$ APs and $K$ UEs are independently and uniformly distributed within a $1\times1\,\text{km}^2$ area with a wrap-around setting \cite{Network}. For ${\bar {\bf{H}} _{mk}}$, we assume that each AP and UE is equipped with a uniform linear array (ULA). The large-scale coefficient $\beta_{mk}$ between AP $m$ and UE $k$ is computed by the COST 321 Walfish-Ikegami model and other parameters are the same as those in \cite{performance}. Moreover, we generate ${{\bf{U}}_{mk,\text{r}}}$ and ${{\bf{U}}_{mk,\text{t}}}$ randomly in the simulations. We also generate ${{\bf{W}}_{mk}}$ randomly with one strong transmit eigendirection capturing the dominant power. Each UE transmits with the same power $200\,\text{mW}$. Each coherence block consists of ${\tau _c} \!=\! 200$ channel uses and $\tau_p \!=\! KN$, unless further specified.

Fig. \ref{tu1} shows the average SE as a function of $M$ for the fully centralized processing (we shortly call it as ``FCP'' in the following) and the LSFD over the MMSE (L-MMSE) combining and the MR combining with the precoding matrix in \eqref{F_precoding} applied. The labels ``$\times$'' generated by the analytical results in \eqref{L3 SE 2} coincide with the curve generated by simulations, which confirms the accuracy of our derived closed-from SE expressions. Also, we observe that the MMSE (L-MMSE) combining outperforms the MR combining for both the FCP and the LSFD, since the MMSE (L-MMSE) combining can efficiently adopt all antennas at each AP to suppress the co-channel interference.

\begin{figure}[t]
\centering
\includegraphics[scale=0.7]{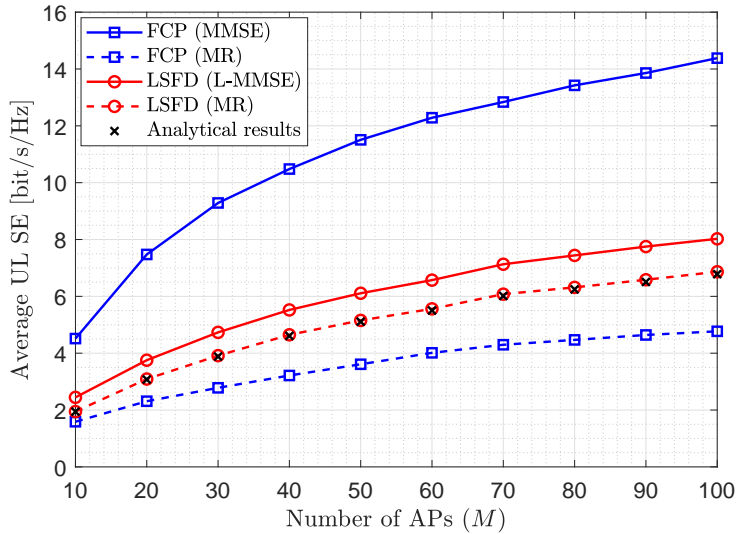}
\caption{Average SE versus the number of AP, $M$, for the LSFD and the FCP with the MMSE (L-MMSE) combining and the MR combining over $K=10$, $L=2$, $N=2$, and $\tau_p = KN/2$.}\label{tu1}
\vspace{-0.4cm}
\end{figure}

\begin{figure}[t]
\centering
\includegraphics[scale=0.7]{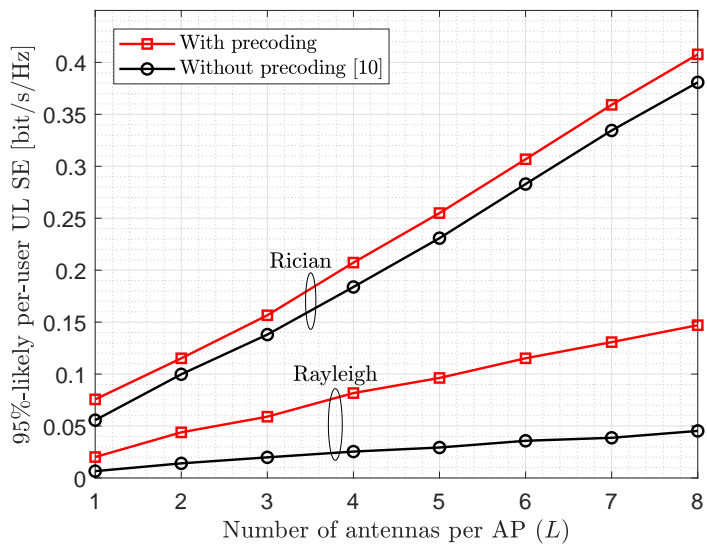}
\caption{95$\%$-likely per-user SE versus the number of antennas per AP, $L$, for the FCP with the MR combining over $M=10$, $K=10$, and $N=4$.}\label{tu2}
\vspace{-0.4cm}
\end{figure}

Fig. \ref{tu2} shows the $95\%\text{-likely}$ per-user SE as a function of $L$ for the FCP with the MR combining over Rician and Rayleigh channels. As observed, the proposed precoding scheme can achieve excellent SE performance in both Rician and Rayleigh channels. Moreover, the SE performance gap between ``with precoding'' and ``without precoding'' in Rician channels is smaller than that of Rayleigh channels (e.g. $13\%$ and $220\%$ for $L=4$, respectively), since the proposed precoding scheme is constructed by exploiting the eigenbasis of correlation matrices of the NLoS component, which only accounts for a relatively small proportion of channel power in Rician channels. Although \eqref{F_precoding} is a heuristic scheme inspired by the mMIMO scenario, this scheme shows the advantages for the implementation of precoding at the UE-side.

\begin{figure}[t]
\centering
\includegraphics[scale=0.7]{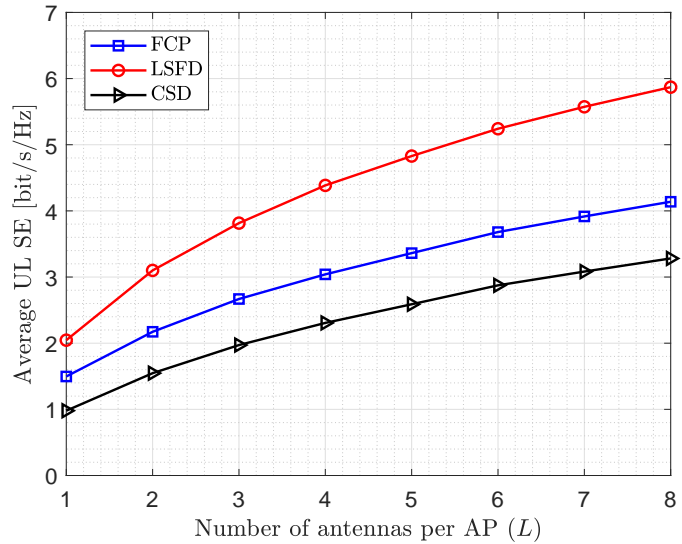}
\caption{Average SE versus the number of antennas per AP, $L$, for the MR combining with $M=20$, $K=10$, and $N=2$.}\label{tu4}
\vspace{-0.2cm}
\end{figure}

Fig. 4 shows the average SE as a function of $L$ for different processing schemes: the FCP, the LSFD, and the classic simple decoding (CSD)\footnote[2]{In the CSD, each AP locally estimates the channels and the CPU performs detection simply with the average of the local estimates from the APs, as proposed in \cite{smallcell}. Known as ``Level 2" in \cite{Making}, the CSD is regarded as a simplification of the LSFD.
} with the MR combining. We observe that the average SE for different processing schemes increases with $L$. Moreover, the average SE for the LSFD is almost twice as that of the CSD. The reason is that compared to the CSD, the CPU in the LSFD linearly weights the local estimates from all APs with the optimal LSFD coefficient matrices to maximize the SE. In particular, when the MR combining is applied, the FCP cannot effectively exploit the advantages brought by the fully centralized implementation. However, the LSFD can achieve excellent SE performance due to its two-layer processing scheme.

\begin{figure}[t]
\centering
\includegraphics[scale=0.7]{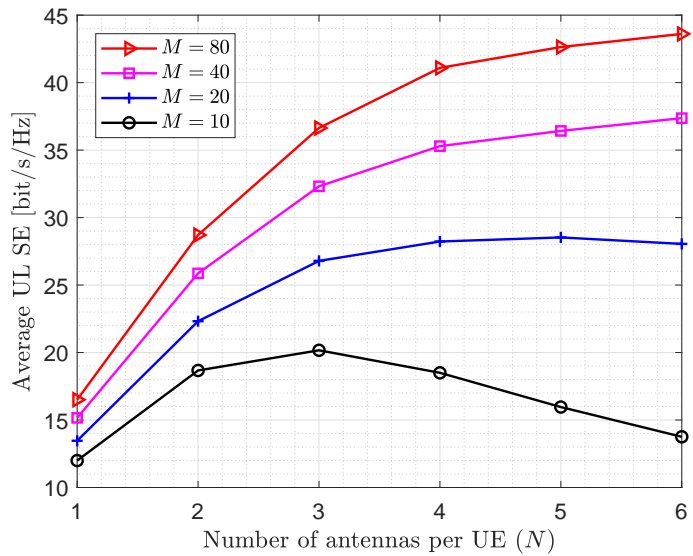}
\caption{Average SE versus the number of antennas per UE, $N$, for the FCP with the MMSE combining over $K=10$ and $L=4$.}\label{tu3}
\vspace{-0.4cm}
\end{figure}

Fig. \ref{tu3} shows the average SE as a function of $N$ for the FCP with the MMSE combining over different $M$. When $M=10$, the average SE increases to its maximum value and then decreases when $N$ increases. In fact, increasing $N$ would significantly reduce the pre-log factor $( \tau _c-\tau _p )/ \tau_c$ in all SE expressions. In particular, for small values of $M$, the decreases incurred by the pre-log factor outweigh the gain in having more UE antennas leading to the SE performance degradation. When $M$ increases, the SE performance can be improved by allowing each UE to be equipped with more antennas. In general, the SE can greatly benefit from a proper number of antennas per UE, especially with a large number of APs. For example, the average SE with $M = 20$, $N = 4$ achieves $110\%$ improvement compared to the single-antenna UE deployment. Moreover, the UEs in practical networks are usually equipped with a few antennas and the computational complexity to optimize the system is acceptable.

\vspace{-5mm}
\section{Conclusions}

We investigated the UL SE performance of CF mMIMO systems over Weichselberger Rician fading channels with phase-shifts and derived novel closed-form SE expressions for the LSFD scheme with the MR combining. Moreover, we proposed an effective UL precoding scheme based on only the eigenbasis of UE-side correlation matrices. The numerical results showed that the proposed precoding scheme is more efficient for Rayleigh channels to improve the SE performance compared with that of Rician channels. Furthermore, the system SE can greatly benefit from a proper number of antennas per UE, especially with a large number of APs.
\vspace{-5mm}
\appendix
\subsection{Proof of Corollary \ref{coro 2}}\label{appendix B}

Based on the theory of the optimal receiver \cite{Zhe_TWC}, the optimal receive combining, maximizing \eqref{L4 SE 1}, is $\mathbf{V}_{k}^{*}=( \mathbf{A}-\mathbf{BB}^H) ^{-1}\mathbf{B}$, where $\mathbf{A}=\sum_{l=1}^K{\mathbf{\hat{H}}_l\mathbf{\hat{F}}_l\mathbf{\hat{H}}_{l}^{H}}\!\!+\!\!\sum_{l=1}^K{\mathbf{C}_{l}^{\prime}}
+\sigma ^2\mathbf{I}_{ML}$ and $\mathbf{B}=\mathbf{\hat{H}}_k\mathbf{F}_k$. Except from having another scaling matrix $\mathbf{I}_N-( \mathbf{B}^H\mathbf{A}^{-1}\mathbf{B}+\mathbf{I}_{N} ) ^{-1}\mathbf{B}^H\mathbf{A}^{-1}\mathbf{B}$, which would not affect the value of \eqref{L4 SE 1}, $\mathbf{V}_{k}^{*}$ is equivalent to the MMSE combining matrix in \eqref{L4 MMSE combing}. So the MMSE combining matrix in \eqref{L4 MMSE combing} can also maximize the achievable SE \cite{Network}.

\subsection{Proof of Theorem \ref{thm_1}}\label{appendix C}

\sloppy
We have ${\bf Z}_{mk}\!=\!{\mathbb E}\{ {\hat {\bf{H}}_{mk}^H{{\hat{\bf{H}}}_{mk}}}\} $ with the $({i,j})$-th element given by $[{\bf Z}_{mk}]_{ij}\!\!=\! {\mathbb E}\{\hat {\bf{h}}_{mki}^H{\hat {\bf{h}}_{mkj}}\}\!\!=\! {\mathrm {tr}}( {{\bf{{\bf{\Phi}}}}_{mk}^{ji} \!+\! {{\bar {\bf{h}} }_{mkj}}\bar {\bf{h}} _{mki}^H})$. We structure ${{{\bf{\Theta}} _{kl}}}\!=\!{\mathbb E}\{ {{\bf{G}}_{kl}}{\hat {\bf{F}}_l}{{\bf{G}}_{kl}^H}  \}$ into blocks, where the $({i,j})$-th element of ${\bf{\Theta }}_{kl}^{mn}$ is
$\!\sum\nolimits_a^N{\sum\nolimits_b^N{{{[{{\hat{\bf{F}}_{l}}} ]}_{ba}}}}{\mathbb E}\{{\hat{\bf{h}}_{mki}^H{{\bf{h}}_{mlb}}{\bf{h}}_{nla}^H{{\hat {\bf{h}}}_{nkj}}}\}$.
Let $\mathbf{\tilde{\Phi}}_{mk}\!=\!\mathbf{\Phi }_{mk}^{1/2}$ and we can also structure $\mathbf{\tilde{\Phi}}_{mk}$ into $N^2$ blocks. So $\hat {\bf{h}}_{mki}^H$ can be rewritten as ${\hat {\bf{h}}_{mki}}\!=\!{\bar{\bf{h}}_{mki}}{e^{j{\varphi_{mk}}}}\!+\!\sum\nolimits_u^N {\tilde {\bf{{\bf{\Phi}}}}_{mk}^{iu}{\bf{x}}_{mk}^u} $ with $\tilde{\bf{{\bf{\Phi}}}}_{mk}^{iu}$ being the $({i,u})$-th submatrix of ${\tilde {\bf{{\bf{\Phi}}}}_{mk}}$ and ${\bf{x}}_{mk}^u \sim {{\mathcal {N}}_{\mathbb C}}({{\bf{0}},{\bf I}_L})$.
For $m = n$, we have $\mathbb{E}\{\mathbf{\hat{h}}_{mki}^{H}\mathbf{h}_{mlb}\mathbf{h}_{mla}^{H}\mathbf{\hat{h}}_{mkj}\}\!=\!\mathbb{E} \{\mathbf{\hat{h}}_{mki}^{H}\mathbf{\hat{h}}_{mlb}\mathbf{\hat{h}}_{mla}^{H}\mathbf{\hat{h}}_{mkj}\}\!+\!\mathbb{E} \{\mathbf{\hat{h}}_{mki}^{H}\mathbf{\tilde{h}}_{mlb}\mathbf{\tilde{h}}_{mla}^{H}\mathbf{\hat{h}}_{mkj}\}$, where the second term is given by ${\mathrm{tr}}({({{{\bar{\bf{h}}}_{mkj}}\bar{\bf{h}}_{mki}^H\!+\!{\bf{{\bf{\Phi}}}}_{mk}^{ji}})\cdot{{\bf{C}}_{ml}^{ba}}})$. If $l = k$, the first term is $\mathbf{\bar{h}}_{mki}^{H}\mathbf{\bar{h}}_{mlb}\mathbf{\bar{h}}_{mla}^{H}\mathbf{\bar{h}}_{mkj}
+\sum\nolimits_{u_1}^N{\sum\nolimits_{u_2}^N{\mathrm{tr}(\mathbf{\tilde{\Phi}}_{mk}^{u_1i}\mathbf{\tilde{\Phi}}_{ml}^{bu_2}
\mathbf{\tilde{\Phi}}_{ml}^{u_2a}\mathbf{\tilde{\Phi}}_{mk}^{ju_1})}}
+ \sum\nolimits_{u_1}^N{[\mathrm{tr}(\mathbf{\bar{h}}_{mki}^{H}\mathbf{\bar{h}}_{mlb}
\mathbf{\tilde{\Phi}}_{ml}^{u_1a}\mathbf{\tilde{\Phi}}_{mk}^{ju_1})
+ \mathrm{tr}(\mathbf{\tilde{\Phi}}_{mk}^{u_1i}\mathbf{\tilde{\Phi}}_{ml}^{bu_1}\mathbf{\bar{h}}_{mla}^{H}\mathbf{\bar{h}}_{mkj}) ]}
+ \sum\nolimits_{u_1}^N{\sum\nolimits_{u_2}^N{\mathbf{\bar{h}}_{mki}^{H}\mathbf{\tilde{\Phi}}_{ml}^{bu_1}
\mathbf{\tilde{\Phi}}_{ml}^{u_1a}\mathbf{\bar{h}}_{mkj}}}$. If $l \in {{\mathcal P}_k}\backslash \{ k \}$, due to the random characteristic of the phase-shifts, we have ${\mathbb E}\{{{{( {{{\bar {\bf{h}} }_{mki}}{e^{j{\varphi _{mk}}}}} )}^H}{{\bar {\bf{h}} }_{mlb}}{e^{j{\varphi _{ml}}}}}\}\!\!=\!\!{\mathbb E}\{{{{({{{\bar{\bf{h}}}_{mla}}{e^{j{\varphi_{ml}}}}})}^H}{{\bar {\bf{h}} }_{mkj}}{e^{j{\varphi _{mk}}}}} \} \!\!=\! 0$ compared to the scenario with $l = k$. If $l \notin {{\mathcal P}_k}$, we derive
${{\mathbb E}}\{ {\hat {\bf{h}}_{mki}^H{{\hat {\bf{h}}}_{mlb}}\hat {\bf{h}}_{mla}^H{{\hat {\bf{h}}}_{mkj}}} \}={\mathrm {tr}}( {( {{{\bar {\bf{h}} }_{mkj}}\bar {\bf{h}} _{mki}^H + {\bf{{\bf{\Phi}}}} _{mk}^{ji}} ) ( {{{\bar {\bf{h}} }_{mlb}}\bar {\bf{h}} _{mla}^H + {\bf{{\bf{\Phi}}}} _{ml}^{ba}} )} )$.
For $m \ne n$, we define ${\mathbb E}
\{{\hat{\bf{h}}_{mki}^H{{\bf{h}}_{mlb}}{\bf{h}}_{nla}^H{{\hat {\bf{h}}}_{nkj}}} \}\!=\!\gamma$. If $l = k$, we obtain $\gamma = {\mathrm {tr}}( {( {{{\bar {\bf{h}} }_{mkb}}\bar {\bf{h}} _{mki}^H + {\bf{{\bf{\Phi}}}} _{mk}^{bi}}){{( {{{\bar {\bf{h}} }_{nka}}\bar {\bf{h}} _{nkj}^H + {\bf{{\bf{\Phi}}}} _{nk}^{aj}} )}^H}})$. If $l \in {{\mathcal P}_k}\backslash \{ k \}$, we derive $\gamma =\sum\nolimits_{u_1}^N {\sum\nolimits_{u_2}^N {( {{\mathrm {tr}} ( {\tilde {\bf{\Phi}} {{_{mk}^{u_1i}}}\tilde {\bf{\Phi}} _{ml}^{bu_1}} )  {\mathrm {tr}} ( {\tilde {\bf{\Phi}} {{_{nl}^{u_2a}}}\tilde {\bf{\Phi}} _{nk}^{ju_2}} )} )} }$.
\vspace{-5mm}
\bibliographystyle{IEEEtran}

\bibliography{IEEEabrv,Ref}

\end{document}